\documentclass[12pt]{article}

\usepackage{header}
\usepackage{tikz}

\title{\Large{ {\bf A simple characterization of single-peaked domains}}}

\author{Mihir Bhattacharya%
\thanks{Department of Economics, Ashoka University, Rajiv Gandhi Education City, Rai, Sonipat, Haryana, India - 131029. Email: \texttt{mihir.bhattacharya@ashoka.edu.in}}\; and Anup Pramanik%
\thanks{Department of Economics, Shiv Nadar Institution of Eminence, NH - 91, Gautam Buddha Nagar, Uttar Pradesh, India - 201314. Email: \texttt{anup.pramanik@snu.edu.in}.}}

\begin{document}

\maketitle
 
\begin{abstract} 

This paper characterizes the single-peaked domain on a tree via the strategy-proofness of extreme rules defined on that tree. For any tree, these rules are unanimous and anonymous on any preference domain. In particular, we show that they are strategy-proof only on the single-peaked domain associated with that tree.

\end{abstract}



\newpage

\section{Introduction}

Single-peaked preferences occupy a central place in political economy and social choice. First formalized by \cite{black1948rationale}, they can be  informally described as follows. Alternatives are arranged along an underlying order. A preference ordering is single-peaked with respect to this order if for any triple of alternatives $a,b$ and $c$ such that $b$ is between $a$ and $c$, either $b$ is preferred to $a$ or $b$ is preferred to $c$. A collection of preference orderings constitutes a single-peaked domain when all preferences are single-peaked with respect to a common underlying order. This restriction emerges naturally in many economic and political contexts. More importantly, it has powerful implications for aggregation. In particular, single-peakedness ensures the existence of well-behaved social choice functions, both in the sense of Arrow’s framework and under strategic considerations (see \cite{Arrow51} and \cite{moulin1980}). While the classical definition of single-peaked preference is based on an exogenous order over alternatives, many economic environments are more naturally represented by networks, with tree structures providing a canonical generalization (see, for example, \cite{demange1982single}, \cite{danilov1994structure} and \cite{schummer2002strategy}). In this paper, our objective is to provide a characterization of the single-peaked domain on a tree.

We characterize the domain via the strategy-proofness of a class of rules, called extreme rules, defined with respect to a fixed tree whose nodes coincide with the set of alternatives. Each extreme rule is associated with a designated leaf of the tree. Given any preference profile, consider the set of top-ranked alternatives and the minimal connected subgraph that contains them. The rule selects the alternative in this subgraph that is closest to the designated leaf. These rules are well-defined for any domain of strict preferences and satisfy unanimity and anonymity. However, whether they are strategy-proof depends on the domain under consideration. Our main result is as follows. Fix a tree and consider the set of all extreme rules on that tree. A domain is single-peaked with respect to this tree if and only if every extreme rule is strategy-proof.

A strand of the literature provides characterizations of the single-peaked domain. \cite{ballester2011characterization} show that a preference profile is single-peaked if and only if it satisfies a collection of restrictions on all triples and quadruples of alternatives. \cite{puppe2018single} characterizes single-peaked domains as Condorcet domains that are minimally rich, connected, and contain a pair of completely reversed preference orders. \cite{chatterji2016characterization} show that a domain is single-peaked on a tree if and only if it admits a strategy-proof, unanimous, tops-only random social choice function satisfying a compromise property. Our result complements this literature by providing a characterization of single-peaked domains on trees via the strategy-proofness of a simple class of rules defined with respect to the underlying tree.

Another strand of the literature seeks to overcome the \cite{gibbard1973manipulation}-\cite{satterthwaite1975strategy} impossibility by fixing a rule with desirable properties and characterizing the domains on which it is strategy-proof; see, for instance, \cite{barbera1998maximal}, \cite{barbie2006non}, \cite{sanver2009strategy}, and \cite{bandhu2022strategy}. In contrast, our approach proceeds in the opposite direction. We characterize the single-peaked domain on a tree via the strategy-proofness of the class of extreme rules defined on that tree.

The remainder of the paper is organized as follows. Section \ref{Model} introduces the model. Section \ref{extreme} introduces the class of extreme rules on a tree and studies their basic properties. Section \ref{MR} presents and proves the main characterization result.

\section{Model}\label{Model}
Let $X$ be a finite set of alternatives with $|X|\geq 3$. The set of agents is $N = \{1, \ldots, n\}$. For each agent $i \in N$,  let $P_i$ denote agent $i$'s preference, which is a strict ordering over the elements of $X$.\footnote{A strict ordering over $A$ is  a \textit{complete}, \textit{transitive} and \textit{antisymmetric} binary relation over $X$.} For any $a,b \in X$, we write $a\;P_i\;b$ if agent $i$ strictly prefers 
alternative $a$ to $b$ under $P_i$. Let $\mathbb{P}$ denote the set of all strict orderings over $X$, referred to as the \emph{unrestricted domain}. An arbitrary subset $\mathbb{D} \subseteq \mathbb{P}$ is called a \emph{domain}. Throughout this paper, we assume that all agents share the same preference domain.

A preference profile $P = (P_1,\ldots,P_n) \in \mathbb{D}^n$ specifies one preference for each agent. For any subset $S \subseteq N$, let $P_S = (P_i)_{i\in S}$ denote the sub-profile of agents in $S$, and write $P=(P_S,P_{-S})$. In particular, for a single agent $i$, we write $P=(P_i,P_{-i})$. For any $P_i \in \mathbb{D}$, let $\tau(P_i)$ denote agent $i$'s top-ranked alternative. For any subset $S \subseteq N$, let $\tau(P_S)=\{\tau(P_i)\mid i\in S\}$ denote the set of top-ranked alternatives of agents in $S$. In particular, for a profile $P \in \mathbb{D}^n$, we write $\tau(P)=\tau(P_N)$.

\begin{defn}\rm
A \emph{social choice function} (scf), or a \emph{rule}, is a mapping $f : \mathbb{D}^n \rightarrow X$.
\end{defn}

The following incentive property of a scf is central in the mechanism design literature and will play a key role in our paper.

\begin{defn}\rm
A scf $f$ is \emph{manipulable} by agent $i$ at a profile $P\in \mathbb{D}^n$ via $P'_i\in \mathbb{D}$ if
\[
f(P'_i,P_{-i})\;P_i\;f(P).
\]
A scf $f$ is \emph{strategy-proof} if it is not manipulable by any agent at any profile.
\end{defn}

We assume that an agent's preference is her private information.  Strategy-proofness, therefore, states that each agent has a weakly dominant strategy to report her preferences truthfully: no agent can benefit from misreporting, regardless of her beliefs about the reports of others.

We also consider the following standard properties of scfs.

\begin{defn}\rm
A scf $f$ satisfies \emph{efficiency} if for every profile $P\in\mathbb{D}^n$, there is no alternative $a\in X$ such that
\[
a\;P_i\;f(P) \quad \text{for all } i\in N.
\]
\end{defn}

Efficiency requires that the rule never selects a Pareto dominated alternative.

\begin{defn}\rm
A scf $f$ satisfies \emph{unanimity} if for every profile $P \in \mathbb{D}^n$ and every alternative $a \in X$, whenever $\tau(P)=\{a\}$, we have
\[
f(P)=a.
\]
\end{defn}

Note that efficiency implies unanimity. A \emph{permutation} of $N$ is a bijection $\sigma:N\to N$. For a profile $P=(P_1,\ldots,P_n)\in\mathbb{D}^n$, define $P^\sigma=(P_{\sigma(1)},\ldots,P_{\sigma(n)})$. 

\begin{defn}\rm A scf $f$ satisfies \emph{anonymity} if for every $P\in\mathbb{D}^n$ and every permutation $\sigma$ of $N$,
\[
f(P^\sigma)=f(P).
\]
\end{defn}

Anonymity requires that the rule treat all agents symmetrically: the outcome depends only on the profile of preferences and not on the identities of the agents.

\section{Extreme Rules on a Tree}\label{extreme}
We introduce a class of scfs that will be central to our analysis. To do so, we first introduce some graph-theoretic notation.
Let $G=(X,E)$ be a graph on $X$, where the node set is $X$ and $E$ consists of unordered pairs of elements of $X$. A \emph{path} from $x$ to $y$ in $G$ is a sequence of distinct nodes $(x_0,x_1,\dots,x_k)$ such that $x_0=x$, $x_k=y$, and
$\{x_{j-1},x_j\}\in E$ for all $j=1,\dots,k$. The graph $G$ is \emph{connected} if for every $x,y\in X$ there exists a path from $x$ to $y$. A \emph{cycle} is a sequence of nodes $(x_0,x_1,\dots,x_k)$ with $k\ge 3$ such that $x_0=x_k$ and $(x_0,x_1,\dots,x_{k-1})$ is a path in $G$. We assume that $G$ is a \emph{tree}, that is, a connected graph with no cycles. Equivalently, in a tree there exists a unique path between any
two nodes in $X$.

For any $x,y\in X$ with $x\neq y$, let $P(x,y)$ denote the set of nodes on the unique path from $x$ to $y$ in $G$. For completeness, define $P(x,x)=\{x\}$ for every $x\in X$.  The \emph{distance} between $x$ and $y$ is defined as $d(x,y)=|P(x,y)|-1$, the number of edges on this path. 

For any non-empty $S\subseteq X$, the \emph{path hull} of $S$ in $G$ is
\[
H(S)=\bigcup_{x,y\in S} P(x,y),
\]
that is, the set of nodes lying on the path between any two nodes in $S$.\footnote{This set is also known as the \emph{Steiner hull} of $S$ in $G$, i.e., the smallest connected sub-graph of $G$ containing $S$. A graph $H=(Y,F)$ is a \emph{sub-graph} of $G$ if $Y\subseteq X$ and $F\subseteq \{\{x,y\}\in E : x,y\in Y\}$.} Since $G$ is a tree, for every non-empty set $S \subseteq X$ and every $x \in X$, the set
\[
\arg\min_{y \in H(S)} d(x,y)
\]
contains a unique node. Denote this node by $\pi(x,H(S))$, the \emph{distance minimizer of $x$ on $H(S)$}. With a slight abuse of notation, we write $\pi(x,S)$ instead of $\pi(x,H(S))$.

Finally, a node $\ell\in X$ is a \emph{leaf} (or \emph{terminal} node) if there exists exactly one node $x\in X$ such that $\{\ell,x\}\in E$. Since $G$ is a tree, it follows that $G$ contains at least two leaf nodes.

\begin{defn}\rm
A social choice function $f : \mathbb{D}^n \rightarrow X$ is called an
\emph{extreme rule on the tree $G$} if there exists a leaf $\ell\in X$
such that for every profile $P \in \mathbb{D}^n$,
\[
f(P)=\pi(\ell,\tau(P)).
\]
\end{defn}

For each leaf $\ell\in X$, let $f^\ell$ denote the extreme rule associated with $\ell$. 
Let
\[
\mathcal{F}(G)=\{f^\ell : \ell \text{ is a leaf of } G\}
\]
denote the class of extreme rules on the tree $G$.

We conclude this section by highlighting several basic properties of extreme rules on a tree.

\begin{remark}\rm
Extreme rules are well-defined for any domain $\mathbb{D} \subseteq \mathbb{P}$. Their definition depends only on the underlying tree $G$ and the reported peaks $\tau(P)$, and therefore does not impose any restriction on the preference domain.
\end{remark}

\begin{remark}\rm
Every extreme rule $f^\ell \in \mathcal F(G)$ satisfies anonymity and unanimity on any domain $\mathbb{D} \subseteq \mathbb{P}$.
Anonymity follows since the rule depends only on the set of reported peaks $\tau(P)$.
Unanimity holds because if $\tau(P)=\{x\}$, then $H(\tau(P))=\{x\}$ and hence $f^\ell(P)=x$.
\end{remark}

\begin{remark}\rm
The efficiency of extreme rules depends on the structure of the underlying tree and the domain of preferences. If the tree is a line, then every extreme rule is Pareto efficient on any domain $\mathbb{D} \subseteq \mathbb{P}$. Indeed, in this case the outcome of the rule at any profile coincides with the top-ranked alternative of some agent, and therefore cannot be Pareto dominated. If the tree is not a line, the efficiency of extreme rules depends on the domain under consideration. In particular, if $\mathbb{D}=\mathbb{P}$ and the tree is not a line, then every extreme rule fails efficiency. We omit the details.
\end{remark}

\section{The Main Result}\label{MR}

This section presents the main result of the paper, which provides a characterization of single-peaked domains on a tree in terms of the strategy-proofness of extreme rules.

A domain $\mathbb{D}$ is \emph{minimally rich} if for every alternative $a \in X$, there exists a preference $P_i \in \mathbb{D}$ such that $a$ is the top-ranked alternative under $P_i$. We now define single-peaked preferences on a tree.

\begin{defn}\rm
A preference $P_i \in \mathbb{P}$ is \emph{single-peaked on the tree $G$} if, for every distinct $a,b \in X$, whenever $b \in P(\tau(P_i),a)$, we have
$b \; P_i \; a$.
\end{defn}

Let $\mathbb{P}(G)$ denote the set of preferences that are single-peaked on the tree $G$.  
A domain $\mathbb{D} \subseteq \mathbb{P}$ is \emph{single-peaked on $G$} if $\mathbb{D} \subseteq \mathbb{P}(G)$.

We are now ready to state the main theorem.

\begin{theorem}\label{thm1}
Let $G$ be a tree on $X$ and let $\mathbb{D}$ be a minimally rich domain. The following statements are equivalent:
\begin{enumerate}
\item Every $f\in \mathcal{F}(G)$ is strategy-proof on $\mathbb{D}$.
\item $\mathbb{D}$ is single-peaked on $G$.
\end{enumerate}
\end{theorem}

\begin{proof} $(1 \Rightarrow 2)$  Let $G$ be a tree on $X$. Assume that the domain $\mathbb{D}$ is minimally rich and that every $f\in \mathcal{F}(G)$ is strategy-proof on $\mathbb{D}$. We show that every $P_i\in\mathbb{D}$ must be single-peaked on $G$.

Suppose, for contradiction, that there exists
$P_i^*\in\mathbb{D}$ that is not single-peaked on $G$.
Let $t=\tau(P_i^*)$. Since $P_i^*$ is not
single-peaked, there exist distinct alternatives $a,b\in X$ such that $b\in P(a,t)$ but $t \; P_i^* \; a \; P_i^* \; b$ .

Choose a leaf $\ell\in X$ such that $a\in P(\ell,b)$. Since $G$ is a tree, such a leaf $\ell$ exists.
By minimal richness of $\mathbb{D}$, there exists a preference
$P'_i\in\mathbb{D}$ such that $\tau(P'_i)=b$. Now consider a profile $P\in\mathbb{D}^n$ defined as follows: $P_1=P_i^*$, and $P_j= P'_i$ for all  $j\neq 1$. Note that $\tau(P)=\{b,t\}$ and hence $H(\tau(P)) = P(b,t)$. Consider a rule $f\in \mathcal{F}(G)$ whose associated leaf is $\ell$. Therefore, we have $f(P)=\pi(\ell,\tau(P))=b$.

Now at $P$, replace agent $1$'s preference by $P'_1$ such that $\tau(P'_1)=a$. By minimal richness of $\mathbb{D}$, such a preference exists. Let the resulting profile be $P'=(P'_1,P_{-1})$. Then $\tau(P')=\{a,b\}$ and hence $H(\tau(P'))=P(a,b)$. Therefore, we have $f(P')=\pi(\ell,\tau(P'))=a$. 

Note that $f(P') \; P_1 \; f(P)$. Hence $f$ is manipulable at $P$ by agent $1$ via $P'_1$, contradicting the assumption that $f$ is strategy-proof on $\mathbb{D}$. Therefore every $P_i \in \mathbb{D}$ must be single-peaked on $G$.

\bigskip

$(2 \Rightarrow 1)$ Let $G$ be a tree on $X$. Assume that the domain $\mathbb{D}$ is minimally rich and is single-peaked on the tree $G$. Consider any  $f\in \mathcal{F}(G)$. Let $\ell$ be the leaf associated with $f$. Then, for every profile $P\in\mathbb{D}^n$, $f(P)=\pi(\ell,\tau(P))$. We show that $f$ is strategy-proof on $\mathbb{D}$.

Fix an agent $i\in N$, a profile $(P_i,P_{-i})\in\mathbb{D}^n$, and a preference $P'_i\in\mathbb{D}$. Let $S=\tau(P_{-i})$.
Define $z=\pi(\ell,S)$, $x=\pi(\ell,S\cup\{\tau(P_i)\})$,  $y=\pi(\ell,S\cup\{\tau(P'_i)\})$. Thus $x=f(P_i,P_{-i})$ and $y=f(P'_i,P_{-i})$.

We first prove the following claim.

\begin{claim}\label{C1}
$x,y \in P(\ell,z)$.
\end{claim}

\begin{proof}
Note that $z$ is the distance minimizer of $\ell$ on $H(S)$. Since $G$ is a tree and $H(S)\subseteq H(S\cup\{\tau(P_i)\})$, the distance minimizer $x$ of $\ell$ on $H(S\cup\{\tau(P_i)\})$ must lie on $P(\ell,z)$. Hence $x\in P(\ell,z)$.

Similarly, since $H(S)\subseteq H(S\cup\{\tau(P'_i)\})$, the distance minimizer $y$ of $\ell$ on $H(S\cup\{\tau(P'_i)\})$ must lie on $P(\ell,z)$. Hence $y\in P(\ell,z)$.
\end{proof}

\medskip

We are now ready to complete the proof. We show that either $x=y$ or $x\;P_i\;y$. This implies that agent $i$ cannot manipulate at $P$ via $P'_i$. If $x=\tau(P_i)$ or $x=y$, there is nothing to prove. Hence assume that $x\neq \tau(P_i)$ and $x\neq y$. We now consider two cases.

\medskip
\textbf{Case 1:} $x\in H(S)$. Since $z=\pi(\ell,S)$, $x=\pi(\ell,S\cup\{\tau(P_i)\})$, and $H(S)\subseteq H(S\cup\{\tau(P_i)\})$, we must have $x=z$. Moreover, since $x\neq \tau(P_i)$, it follows that $\tau(P_i)\notin P(\ell,z)$. 

Since $x\neq y$ and, by Claim~\ref{C1}, $y\in P(\ell,z)$, we have
$y\in P(\ell,z)\setminus\{z\}$. Hence $z\in P(\tau(P_i),y)$ and therefore
\[
x\in P(\tau(P_i),y).
\]
Since $P_i$ is single-peaked on $G$, it follows that $x\;P_i\;y$.

\medskip

\textbf{Case 2.} $x\notin H(S)$. We first show that $\tau(P_i)\notin H(S)$. Suppose, to the contrary, that $\tau(P_i)\in H(S)$. Then $H(S)=H(S\cup\{\tau(P_i)\})$, which implies $z=x$. Hence $x\in H(S)$, a contradiction.

Let $s$ be the distance minimizer of $\tau(P_i)$ on $H(S)$, that is, $s=\pi(\tau(P_i),S)$. Since $x=\pi(\ell,S\cup\{\tau(P_i)\})$ and $x\notin H(S)$, it follows that $s=z$.

Let $(a_1=\tau(P_i),a_2,\ldots,a_k=z)$ be the path from $\tau(P_i)$ to $z$. Since $x\notin H(S)$ and $x\neq \tau(P_i)$, it must be that
\[
x\in P(\tau(P_i),z)\setminus\{\tau(P_i),z\}.
\]
Let $x=a_j$ for some $1<j<k$.

Let $(b_1=\ell,b_2,\ldots,b_m=x)$ be the path from $\ell$ to $x$. Since $x,y\in P(\ell,z)$ by Claim~\ref{C1}, we have
\[
y\in \{b_1,\ldots,b_m=x,a_{j+1},\ldots,a_k=z\}.
\]
Because $y\neq x$, it follows that $x\in P(\tau(P_i),y)$. Hence
$x\;P_i\;y$ by the single-peakedness of $P_i$.

\medskip

Therefore, in all cases, either $x=y$ or $x\;P_i\;y$. Since $i$, $(P_i,P_{-i})$, and $P'_i$ were arbitrary, $f$ is strategy-proof on $\mathbb{D}$.
\end{proof}

We conclude with several remarks concerning the theorem.

\begin{remark}\rm
The minimal richness assumption in Theorem~\ref{thm1} is made only for ease of exposition. Without this assumption, let
\[
X'=\{x\in X \mid x=\tau(P_i)\text{ for some }P_i\in\mathbb{D}\}
\]
denote the set of alternatives that arise as the top-ranked alternative in some preference in $\mathbb{D}$. We assume that $|X'|\ge 3$. For any preference $P_i$, let $P_i|_{X'}$ denote the restriction of $P_i$ to $X'$, defined by
\[
a\, P_i|_{X'}\, b \quad \text{if and only if} \quad a\, P_i\, b
\quad \text{for all } a,b\in X'.
\]
Let
\[
\mathbb{D}|_{X'}=\{P_i|_{X'} \mid P_i\in\mathbb{D}\}
\]
denote the restriction of $\mathbb{D}$ to $X'$. Then the theorem can be restated as follows: for any tree $G'$ on $X'$, the domain $\mathbb{D}|_{X'}$ is single-peaked on $G'$ if and only if every $f\in \mathcal{F}(G')$ is strategy-proof on $\mathbb{D}|_{X'}$.
\end{remark}

\begin{remark}\rm
If $\mathbb{D}$ is single-peaked on a tree $G$, then every extreme rule on $G$ is efficient on $\mathbb{D}$. Indeed, for any profile $P\in\mathbb{D}^n$, the alternative selected by an extreme rule lies in $H(\tau(P))$, and $H(\tau(P))$ coincides with the set of efficient alternatives at $P$.
\end{remark}

\begin{remark}\rm
\cite{schummer2002strategy} characterize the class of strategy-proof and onto rules on trees as extended median voter schemes (e.m.v.s.), which are defined through a family of generalized median voter schemes satisfying a consistency condition across paths. When restricted to domains that contain all single-peaked preferences on a tree, their characterization applies to our setting. In particular, every extreme rule is an e.m.v.s., and hence the class of extreme rules forms a subclass of the e.m.v.s. class.
However, the representation of e.m.v.s. is implicit. In contrast, extreme rules admit a simple and explicit representation: the outcome is obtained as the distance minimizer of a fixed leaf over the path hull of reported peaks. Thus, while extreme rules are a special case of e.m.v.s., they provide a simple and direct description of a subclass of strategy-proof rules on a tree.
\end{remark}

\bibliographystyle{ecta}
\bibliography{peaked}

\end{document}